\documentclass{article}

\usepackage{packages}
\usepackage{macros}

\title{Projective Presentations for Lex Modalities}
\author{Mark Damuni Williams}

\begin{document}

\maketitle

\begin{abstract}
  Modalities in homotopy type theory are used to create and access subuniverses of a given type universe. 
  These have significant applications throughout mathematics and computer science, and in particular can be used to create universes in which certain logical principles are true.
  In this work we define presentations of topological modalities, which act as an internalisation of the notion of a Grothendieck topology.
  A specific presentation of a modality gives access to a surprising amount of computational information, such as explicit methods of determining membership of the subuniverse via internal sheaf conditions.
  Furthermore, assuming all terms of the presentation satisfy the axiom of choice, we are able to describe generic and powerful computational tools for modalities.
  This assumption is validated for presentations given by representables in presheaf categories.
  We deduce a local choice principle, and an internal reconstruction of Kripke-Joyal style reasoning.
  We use the local choice principle to show how to relate cohomology between universes, showing that a certain class of abelian groups has cohomolgoy stable between universes.
  We apply the methods to a prominent example, a type theory axiomatising the classifying topos of an algebraic theory, which specialises to give type theories 
  for synthetic algebraic geometry and synthetic higher category theory. 
  We apply the sheaf conditions to show that several presentations of interest are subcanonical, and apply the cohomology methods to show that 
  quasi-coherent modules have cohomology stable between the Zariski, \'etale and fppf toposes.
\end{abstract}

\section{Introduction}

Homotopy type theory (HoTT) is the internal language of $\infty$-topoi.
This means it can effectively be used to describe the synthetic, internal language of different mathematical theories.
For instance Homotopy type theory can be extended via axioms to represent a category for synthetic algebraic geometry~\cite{cherubini2023foundation}, synthetic differential geometry~\cite{myers2022orbifolds} 
or even higher category theory~\cite{gratzer2024directedunivalencesimplicialhomotopy}.
To verify that these axiomatisations are correct, models must be established in certain higher topoi.
This process can be technically difficult, even when it is clear a model does exist.
In recent work a partially internal method of modelling was used~\cite{moeneclaey2024sheaves,cherubini2023foundation}, the idea being to axiomatise a type theory for a presheaf topos, and then to carve out the desired category of sheaves internally.

Working this way has several benefits.
First it can simplify much of the external work, since checking axioms modelled in a presheaf category is often much less technically challenging than checking axioms in a sheaf category.
In our application section we are able to apply strong generic facts about the internal logic of presheaf categories to gain several axioms ``for free''.
Proceeding to carry out the proofs in the sheaf category often becomes easy, and far closer to ordinary logic with modal operators, once we use the internal logic.
Secondly, it can be useful if there are features of a larger universe which it would be helpful to have around.
For instance in the work of~\citeauthor{gratzer2024directedunivalencesimplicialhomotopy}~\cite{gratzer2024directedunivalencesimplicialhomotopy}, it is necessary to work in the category of presheaves of bounded distributive lattices, rather than that of simplicial sets, in order for a certain adjunction to exist.
Carving out of the intended objects, the simplicial sets, can be described via the method here.

In our general theory, we introduce presentations of modalities, to describe presentations of Grothendieck topologies externally.
From these we can describe covers and (higher) sheaf conditions internally.
Further, modifying the presentations to more closely resemble Grothendieck topologies lets us explicitly describe certain
sheafifications, recovering a Kripke–Joyal type semantics internally.

With this theory developed, we apply the method to synthetic algebraic geometry, and synthetic higher category theory.
In the realm of synthetic higher category theory we are able to remove dependencies of two axioms on a result from~\cite{gratzer2024directedunivalencesimplicialhomotopy}, showing that the interval $\mathbb{I}$ is simplicial.
We do this by applying our sheaf conditions to the simplicial modality, and in doing so, reduce the question to pure algebra.

\section*{Acknowledgements}
I want to thank Hugo Moeneclaey for providing much of the initial idea for this work in the special case of synthetic algebraic geometry. 
In particular his internal notion of topology makes the basis for the notion of presentation, and many initial results are his.
I also want to thank my supervisor Ulrik Buchholtz for providing insightful help and guidance throughout this research.

\section{Background}

We begin with a brief recounting of some external topos theory, to justify the definitions in the upcoming section.

In ordinary (1-)topos theory, a Grothendieck topology on a category $\mathcal{C}$ is given by, for each $c \in \mathcal{C}$, a collection of \textbf{sieves} on $c$, which are 
collections of arrows $\{d_i \to c\}_{i : I}$ which are described as ``covering'' $c$.
We then enforce this covering idea by giving the sheaf condition. 
Writing $\yo : \cC \to \mathrm{Psh}(\cC)$ for the Yoneda embedding, the sheaf condition is written as follows:

\begin{definition}
\label{def:sheaf}
    A presheaf $F$ on $\mathcal{C}$ is a sheaf if for all covers $\{d_i \to c\}_{i : I}$ the following diagram is an equaliser:
    \[ F(c) \to \prod_{i : I} F(d_i) \rightrightarrows \prod_{i, j :I}  \hom(\yo_{d_i} \times_{\yo_c} \yo_{d_i}, F) \]
\end{definition}
This intuitively says given a collection $x_i : F(d_i)$ which agree on the intersection of $d_i$ and $d_j$, there is a unique amalgam $x : F(c)$ built up of the $x_i$.

This condition can be described internally to the presheaf category by saying that sheaves $F$ believe the map $\bigsqcup_{i : I} \yo_{d_i} \to \yo_{c}$ to be a surjective map.
In the subcategory of sheaves, the (sheafification of) the previous map \emph{is} an epimorphism.
Furthermore, internally this condition can be given entirely fibrewise, since a map is surjective if and only if all its fibres are inhabited.
Specifically we ask that $\bigsqcup_{i : I} \yo_{d_i} \times_{\yo_{c}} 1$ is inhabited for all maps $1 \to \yo_{c}$. 

Thus, a Grothendieck topology can be described by a collection of objects $T$ which we would like to force to be inhabited.
The additional properties of being a Grothendieck topology show that the terminal object is in $T$ and that $T$ is closed under coproducts. 
This motivates our internal definition of a presentation, subsuming the definition of a Grothendieck topology. 

We now work in HoTT, and use the conventions from the HoTT book \cite{hottbook}, using the univalence axiom and its consequences freely.
We begin with some type theoretic background, developing enough of the general theory of modalities from \cite{spitters2020modalities} for our purposes.

\begin{definition}
    A \textbf{reflective subuniverse} is formed of a subuniverse $\UU_\OO$ of $\UU$, a reflector $\OO : \UU \to \UU_\OO$, and a unit $\eta : \prod_{X : \UU} X \to \OO X$.
    These satisfy \textbf{modal recursion} for all $X : \UU$ and for all $Y : \UU_\OO$ the map 
    \[ - \circ \eta_X : (\OO X \to Y) \to (X \to Y) \]
    is an equivalence.
\end{definition}

Reflective subuniverses can be shown to have many stability properties. 
In particular, they are closed under $\Pi$, $+$, $\times$, and taking equality types.
However, they do not quite constitute logical subuniverses as they are not necessarily stable under $\Sigma$.

\begin{definition}
    A \textbf{modality} is a reflective subuniverse $(\UU_\OO, \OO, \eta)$ satisfying one of the following equivalent conditions:
    \begin{enumerate}
        \item $\UU_\OO$ is closed under $\Sigma$.
        \item $\OO$ satisfies \textbf{modal induction}. That is for all $B : \OO A \to \UU_\OO$, the natural map 
              \[ \prod_{x : \OO A} B(x) \to \prod_{x : A} B(\eta(x))  \]
              is an equivalence.
    \end{enumerate}
\end{definition}

We will often suppress much of the notation when defining a modality. 
We call a type $\textbf{modal}$ if the map $\eta_X$ is an equivalence.
Note that $\UU_\OO$ can be recovered as $\Sigma_{X : \UU} \text{is-modal}(X)$.

There are many natural examples of modalities, for instance, the $n$-truncation modality.
Truncation is an instance of nullification, which will be the basis for all modalities we consider.

\begin{definition}
    Given a type family $B : A \to \UU$, a type $X$ is called \textbf{$B$-null}
    if for all $a : A$ the diagonal map 
    \[ X \to (B(a) \to X) \]
    is an equivalence.
    That is every map $B(a) \to X$ is constant.
\end{definition}

\begin{definition}
    The \textbf{nullification} of $X$ at $B : A \to \UU$ is a type $\OO_B X$ and a natural map 
    $X \to \OO_B X$ such that $\OO_B X$ is $B$-null and so that the natural map satisfies $B$-null recursion.
\end{definition}

The nullification operation exists for any given type family by a higher inductive type construction, 
by taking a type $X$ and manually adds a proof that every map from $B(a) \to X$ is constant.

\begin{lemma}
    Nullification forms a modality.
\end{lemma}

A modality given by nullification at a family of types is called \textbf{accessible}.

\begin{example}
 $n$-truncation is given by nullification at the $(n+1)$-sphere $S^{n+1}$.
\end{example}

Even modalities don't correspond to the external notion of subtopos, as the $n$-truncation modality shows.
For any $n$ the subuniverse of $n$-truncated types is not a sub-$(\infty,1)$-topos of the category of types: an extra condition is needed for that to be the case.

\begin{definition}
    A \textbf{lex modality} is a modality $\OO$ which preserves pullbacks.
\end{definition}

Lex modalities really do correspond to the external notion of subtopos and as such they behave well with respect to many of the structures in HoTT.
For instance, they preserve truncation levels~\cite{spitters2020modalities}, and the universe of modal types is itself modal.
Somewhat surprisingly, however, subtopoi of $(\infty,1)$-topoi don't all come from Grothendieck topologies, unlike in the $1$-topos theoretic case.
The class which do are called topological, and also admit an internal characterisation.

\begin{definition}
    A modality is called \textbf{topological} if it is given by nullification at a family of propositions.
\end{definition}

\begin{lemma}[\cite{spitters2020modalities}]
    Any topological modality is lex.
\end{lemma}

\section{Presentations of Modalities}

Now that we have the correct class of modalities, we can now begin to define our notion of presentation.
First we will call a function $T : I \to \UU$ a \textbf{collection of types}, and we will write $X \in T$ for $\exists_{i : I} T(i) = X$.
We will call $T$ a small collection if $I : \UU$, and large otherwise.

\begin{definition}
    A \textbf{presentation} of a lex modality is a collection $T$ of types with $1 \in T$ and $T$ closed under $\Sigma$.
    That is if $X \in T$ and $B : X \to \UU$ has $B(x) \in T$ for $x : X$ then $\sum_{x : X} B(x) \in T$.
\end{definition}

We can take any collection of types, and generate a presentation by simply adding $1$ and taking its $\Sigma$-closure. 
Given a collection of types $T$ we will denote the presentation it generates by $\langle T \rangle$.

\begin{definition}
    For any collection of types $T$, the modality $\OO_T$ it generates is given by nullifying at the family $P : I \to \Prop_\UU$ with $P(i) = \Trunc{T(i)}$. 
    We call the modal types for this modality \textbf{$T$-sheaves}, and the modality $\OO_T$ the \textbf{$T$-sheafification} modality.
\end{definition}

This modality is thus descending to the subuniverse in which all the types in the presentation are inhabited. 
This appears naturally, for instance in algebraic settings, where we can take the class generated by all solution sets of polynomials. 
Forcing these are all inhabited is something like taking an algebraic closure.

\begin{example}
    Any topological modality has a presentation, generated by the collection of propositions at which we are nullfying. 
    However, for many modalities this isn't a natural choice of presentation, and the choice of specific presentation will give us information about our modality.
\end{example}

Note presentations are far from unique, but for any given modality there is a unique large maximal presentation. 
It is large since the collection of types is indexed over the universe $\UU$.

\begin{definition}
    Given a lex modality $\OO$ the \textbf{maximal presentation} $T_{\mathrm{max}}$ 
    is given by the types $X$ such that $\OO \Trunc{X}$.
\end{definition}

\begin{lemma}
    Let $\OO$ be a topological modality.
    Then $T_{\mathrm{max}}$ is a presentation of $\OO$, and any other presentation is contained within it.
\end{lemma}
\begin{proof}
    First we show it is a presentation of a modality. 
    Clearly $\OO \Trunc{1} = 1$ holds, and so $1 \in T_{\mathrm{max}}$.
    Suppose we have $Y : X \to T_{\mathrm{max}}$ a family of types in the presentation, with $X$ in the presentation.
    Then 
    \[\OO \Trunc*{\sum_{y : Y} X(y) } = \OO \Trunc*{\sum_{y : Y} \OO \Trunc*{X(y)}} = \OO \Trunc*{\sum_{y : Y} 1} = \OO \Trunc*{Y} = 1\]
    so that $\sum_{y : Y} X(y)$ is also in the presentation.

    Now we show that the modalities $\OO$ and $\OO_T$ coincide.
    Since $\OO$ is topological, let $B : A \to \Prop$ be a presentation of it.
    Suppose $X$ is a sheaf for $T_\mathrm{max}$.
    Note by definition, since $\OO B(a)$ holds in nullification at $B$, we have $B(a) \in T_\mathrm{max}$.
    So in particular $X$ is a sheaf for $B$.
    
    Suppose now that $X$ is a sheaf for $\OO$.
    Then for all $Y \in T_\mathrm{max}$ we have, by modal recursion
    \[ X^{\Trunc{Y}} \simeq X^{\OO \Trunc{Y}} \simeq X \]

    The presentation $T_{\mathrm{max}}$ is clearly maximal, since any presentation $T'$ of $\OO$ satisfies $\OO \Trunc{X}$ for $X \in T'$.
\end{proof}

\begin{remark}
    Note this maximal presentation isn't necessarily good.
    It is by construction large, which means it cannot be used to define a modality which preserves universe levels.
\end{remark}

Recall that modalities are naturally ordered, with $\OO \leq \lozenge$ if every $\OO$-modal type is $\lozenge$-modal.
It is clear that if $T_1 \subseteq T_2$ we have $\OO_{T_2} \leq \OO_{T_1}$.
We can more generally compare presentations by a simple criterion: 

\begin{lemma}
    Suppose $T_1$ and $T_2$ are any collections of types, with the property that for all $X \in T_1$ we have $\OO_{T_2} \Trunc{X}$.
    Then $\OO_{T_2} \leq \OO_{T_1}$.
\end{lemma}
\begin{proof}
    Suppose $F$ is a $T_2$ sheaf. 
    Then for all $X \in T_1$ we have $F^{\Trunc{X}} \simeq F^{\OO_{T_2} \Trunc{X}} \simeq F$ by the sheaf condition.
    Hence $F$ is a $T_1$ sheaf.
\end{proof}

\begin{corollary}
    Any collection of types $T$ and its associated presentation $\langle T \rangle$ present the same modality.
\end{corollary}
\begin{proof}
    Since $T \subseteq \langle T \rangle$ we have $\OO_{\langle T \rangle} \leq \OO_{T}$.
    
    Now we work inductively on the structure of $X \in \langle T \rangle$.
    Note if $X = 1$ then clearly $\OO_{T} 1 = 1$.
    If $X \in T$ then also clearly $\OO_{T} \Trunc{X}$.
    Now if $X = \sum_{a : A} B(a)$ where $\OO_{T} \Trunc{A}$ and for all $a : A$ we have $\OO_{T} \Trunc{B(a)}$ then we may calculate
    \[ \OO_{T} \Trunc{ \sum_{a : A} B(a) } = \OO_{T} \Trunc*{ \sum_{a : A} \OO_{T} \Trunc*{ B(a) } } = \OO_{T} \Trunc*{ \sum_{a : A} 1 } = 1\]
    Hence we deduce that for all $X \in \langle T \rangle$ we have $\OO_{T} \Trunc{X}$ so that by the previous lemma, $\OO_{T} \leq \OO_{\langle T \rangle}$ and thus they are equal. 
\end{proof}

Given that taking $\Sigma$-closure and including $1$ doesn't change which modality is presented, one might wonder why we have considered it at all.
The benefit we gain from doing this is covers. 
The notion of a cover of a type only really makes sense when we have closure under $\Sigma$, since otherwise they do not compose.
This will be useful for recovering sheaf conditions, and describing certain covering properties in the upcoming sections.

Using the definition of a presentation, we can recover the notion of a cover in a Grothendieck topology.
Presentations should be thought of as giving the fibres of covering families.
Thus maps whose fibres are all in the presentation are covers.

\begin{definition}
    A \textbf{$T$-cover} is a map $f : X \to Y$ such that for all $y : Y$, we have $\fib_f(y) \in T$.
\end{definition}

Note all $T$-covers become surjective in the subuniverse.
Covers behave as expected for Grothendieck topologies:

\begin{lemma}
    Covers are closed under composition and pullback.
    Any equivalence is a cover.
\end{lemma}
\begin{proof}
    Let $f : X \to Y$ and $g : Y \to Z$ be $T$-covers. 
    Then for any $z : Z$ we have 
    \begin{align*}
        \fib_{g \circ f}(z) 
        &:= (x : X) \times (g(f(x)) = z) \\
        &= (x : X) \times (y : Y) \times (f(x) = y)\times (g(y) = z) \\
        &= ((y, p) : \fib_g(z)) \times \fib_f(y)
    \end{align*}
    Since $\fib_g(z)$ and $\fib_f(y)$ are in $T$, and $T$ is closed under $\Sigma$, we have that this type is in $T$.

    If $f : X \to Z$ is in $T$ and $g : Y \to Z$ then their pullback $\pi : X \times_Z Y \to Y$ of $f$ along $g$ has fibres $\fib_\pi(y) = \fib_f(g(y))$ which are in $T$ by assumption. 

    Finally since $1 \in T$ any equivalence is a $T$-cover, since all their fibre are contractible.
\end{proof}

\section{Sheaf Conditions}

We can now show that sheaves for a presentation satisfy a sheaf condition, justifying the name.
This is the usual sheaf condition for Grothendieck topoi, at least in the case of $0$-types.
More generally we have internal ``stack'' conditions, for $n$-types.
However, the usual formulation of stack conditions externally does not translate directly into type theory, as this would 
require writing out coherence data at arbitrary levels, a known open problem in HoTT. 
Instead we sidestep the issue by using iterated joins, a notion from homotopy theory, to avoid the coherence.
For each external natural number, we gain an equivalence to the usual formulation, but not quantifying internally over the natural numbers.
We recall the definition of joins in HoTT.

\begin{definition}
    The \textbf{join} of two types $A$ and $B$ is the pushout
    \[\begin{tikzcd}
        A \times B \ar[r] \ar[d] & B \ar[d] \\
        A \ar[r] & A \ast B
    \end{tikzcd}\]
\end{definition}

We can extend this notion to the join of maps. Given $f : A \to X$ and $g : B \to X$ we define their join as the pushout of the pullback:
\[\begin{tikzcd}
    A \times_X B \ar[r] \ar[d] &  B \ar[d] \\
    A \ar[r] & A \ast_X B
\end{tikzcd}\]
With this we have an induced map $f \ast g : A \ast_X B \to X$.

\begin{remark}
    Note that $f \ast g$ is \emph{not} the functorial action of $\ast$ on $f$ and $g$.
\end{remark}

\begin{lemma}[\cite{rijke2017join} Theorem 2.2]
    For any $f : A \to X$, $g : B \to X$ we have 
    \[ \fib_{f \ast g}(x) = \fib_{f}(x) \ast \fib_g(x)\]
\end{lemma}

The join allows for an explicit description of the propositional truncation, and more generally the image of any map, other than the standard higher inductive definition, by using a colimit of iterated joins. 
Given a type $A$ we define $A^{\ast n}$ to be the $n$-fold iterated join, and given a map $f : A \to B$ we define its $n$-fold iterated join as $f^{\ast n} : A^{\ast_B n} \to B$.
Note that by induction we have $\fib_{f^{\ast n}}(b) = \fib_f(b)^{\ast n}$.

\begin{lemma}[\cite{rijke2017join}]
    For any $f : A \to B$ we have 
    \[ \im(f) = \colim(A \to A \ast_B A \to A \ast_B A \ast_B A \to \cdots) \]
\end{lemma}

\begin{remark}
    In particular using the unique map $f : A \to 1$ we have
    \[ \| A \| = \colim(A \to A \ast A \to A \ast A \ast A \to \cdots) \]
\end{remark}

The idea is that by repeatedly joining a type with itself, we remove homotopy information from below, one step at a time.
In the limit there is no homotopy information in the sense that all pairs of points are actually equal.

\begin{lemma}
\label{lem:join_to_n_type}
    Let $A$ be a type and $X : A^{\ast(n+3)} \to \UU$ such that for all $a:A^{\ast(n+3)}$, $X(a)$ is an $n$-type. Then the inclusion $\inr : A^{*(n+2)} \to A^{*(n+3)}$ induces an equivalence
    \[ \prod_{z : A^{\ast(n+3)}}X(z) = \prod_{z : A^{\ast(n+2)}}X(\inr(z)) \]
\end{lemma}
\begin{proof}
    By induction on $n$. For $n=-2$ this is trivial, since both sides are the unit type.
    Now we define an inverse to the natural map
    \[ \prod_{z : A^{\ast(n+3)}}X(z) \to \prod_{z : A^{\ast(n+2)}}X(\inr z).\]
    Let $\phi : \prod_{z : A^{\ast(n+2)}}X(\inr z)$. 
    By the universal property of the pushout, we need to supply maps $\prod_{a : A} X(\inl a)$ and $\prod_{z : A^{*(n+2)}} X(\inr z)$ and the gluing data.

    For the first map we choose $\phi \circ \inl$ where $\inl : A \to A^{\ast(n+2)}$ and for the second we choose $\phi$.
    It only remains to show the gluing, that is:
    \[ \prod_{a : A} \prod_{z : A^{\ast(n+2)}} \phi(\inl a) = \phi(z)\]
    Since $X$ is an $n$-type, $\phi(\inl a) = \phi(z)$ is an $(n-1)$-type, so that we may apply induction, to see that this type is equivalent to:
    \[ \prod_{a : A} \prod_{z : A^{\ast(n+1)}} \phi(\inl a) = \phi(\inr z)\]
    But this is satisfied, since for any $a : A$ and $z : A^{\ast(n+1)}$ we have $\inl a = \inr z$ in $A^{\ast(n+2)}$ by gluing data of the pushout.
    This defines the map as required, and it is definitionally a section to the natural map.
    So we only need to check the other round trip. 
    
    Let $\psi : \prod_{z : A^{\ast(n+3)}} X(z)$. 
    This gets mapped to $\psi \circ \inr : \prod_{z : A^{\ast(n+2)}} X(\inr(z))$ which in turn gets mapped to the term of $ \prod_{z : A^{\ast(n+3)}} X(z) $ given by the data of
    \[ \psi \circ \inr \circ \inl : \prod_{a : A} X( \inr(\inl(a)) )\]
    \[ \psi \circ \inr : \prod_{a : A} X( \inr(a) ) \]
    and the gluing term
    \[ \prod_{a : A} \prod_{z : A^{\ast(n+2)}} \psi(\inr(\inl(a))) = \psi(\inr(z))\]
    But note that we have a term $p : \inr \circ \inl = \inl : A \to A^{\ast(n+3)}$ and applying this $p$ to the gluing term gives the standard proof that $\psi \circ \inl = \psi \circ \inr$.
    This means that this map $\prod_{z : A^{\ast(n+3)}} X(z)$ is equal to $\phi$, and we are done.
\end{proof}

Note this lemma easily extends to include families depending on $A^{\ast(n+k)}$ for $k \geq 3$.

\begin{corollary}
      $X^{\Trunc{A}} = X^{A^{\ast(n+2)}}$ for any type $A$ and $n$-type $X$.
\end{corollary}
\begin{proof}
    We have $\Trunc{A} = \colim(A \to A\ast A \to A \ast A \ast A \to \cdots)$ and so 
    \[ X^{\Trunc{A}} = \lim( X^A \leftarrow X^{A \ast A} \leftarrow X^{A \ast A \ast A} \leftarrow \cdots)\]
    But using \cref{lem:join_to_n_type} we have that this sequence stabilises at $n+2$ giving the result.
\end{proof}

We can use this to obtain a sheaf condition. 
To pass from arbitrary $T$-covers to simply elements of the topology, we use a lemma about totalisation of a family of maps to work fibrewise.

\begin{lemma}
    \label{lem:fibrewise_precomp}
    Let $B, C : A \to \UU$ be type families over $A$, and let $f : \prod_{a : A} B(a) \to C(a)$.
    We define the totalisation of $f$ to be 
    \[ \tot_f : \sum_{a : A} B(a) \to \sum_{a : A} C(a) \]
    \[ \tot_f(a, b) = (a, f(b)) \]
    Fix a type $X$. Suppose for all $a : A$, the precomposition map  
    \[ {f_a}^* : (C(a) \to X) \to (B(a) \to X) \] is an equivalence. 
    Then the map
    \[ \psi : X^{\sum_{a : A} C(a)} \to X^{\sum_{a : A} B(a)}\]
    given by composition with $\tot_f$ is an equivalence.
\end{lemma}
\begin{proof}
    Suppose precomposition by $f_a$ is an equivalence for all $a : A$.
    Call the inverse map $\mathrm{ext}_a : X^{B(a)} \to X^{C(a)}$.
    We define the inverse to precomposition with totalisation by
    \[ \phi : X^{\sum_{a : A} B(a)} \to X^{\sum_{a : A} C(a)} \]
    \[ g \mapsto \lambda (a,c) . \mathrm{ext}_a(g_a)(c)\]
    This gives an inverse since for any $p : \sum_{a : A} B(a) \to X$ and any $a : A, c : C(a)$ we have
    \begin{align*}
        \phi(\psi(g))(a, c) &= \phi( g \circ \tot_f)(a,c)
        \\                  &= \mathrm{ext}_a( (g \circ \tot_f)_a )(c) 
        \\                  &= \mathrm{ext}_a(g_a \circ f_a)(c)
        \\                  &= g_a(c)
    \end{align*}

    Similarly, for any $q : \sum_{a : A} C(a) \to X$ and any $a : A, b : B(a)$ we have
    \begin{align*}
        \psi(\phi(q))(a,b) &= (\phi(q) \circ \tot_f)(a,b)
        \\                 &= \phi(q)(a, f_a(b))
        \\                 &= \mathrm{ext}_a(q_a)(f_a(b))
        \\                 &= q_a(b) 
    \end{align*}
\end{proof}

\begin{corollary}[Sheaf Condition] 
\label{cor:sheaf_cond}
    Let $n \geq 0$.
    An $(n-2)$-type $X$ is a $T$-sheaf iff for any $T$-cover $f : A \to B$ the natural map $(B \to X) \to (A^{*_Bn} \to X)$
    is an equivalence.
\end{corollary}
\begin{proof}
($\Rightarrow$)
Suppose $X$ is a $T$-sheaf and $f : A \to B$ is a $T$-cover. 
Then note that $A^{*_B n} = \Sigma_{b : B} \fib_{f^{*n}}(b) = \Sigma_{b : B} \fib_f(b)^{*n}$ and that the map $(B \to X) \to (A^{*_B n} \to X)$ is given by precomposition of the totalisation of the fibrewise maps.
Hence by \cref{lem:fibrewise_precomp} we have that $(B \to X) \to (A^{*_B n} \to X)$ is an equivalecne if for each $b : B$ that diagonal map
\[ X \to (\fib_{f}(b)^{*n} \to X) \]
is an equivalence.
But note that $\fib_{f}(b)^{*n} \to X = \| \fib_{f}(b) \| \to X$ as $X$ is an $(n-2)$-type and since $X$ is a $T$-sheaf and $\fib_f(b) \in T$ as $f$ is a cover, we have that the diagonal map $X \to (\| \fib_f(b) \| \to X)$ is an equivalence.

($\Leftarrow$) Note that for any $X \in T$ the map $f : A \to 1$ is a $T$-cover trivially. Hence the sheaf condition gives the natural map $X \to (A^{*n} \to X)$ is an equivalence. But since $X$ is an $(n-2)$-type this is equivalent to the natural map $X \to (\| A \| \to X)$ being an equivalence.
\end{proof}

\begin{remark}
    Note we have stated this for covers and presentations.
    But an intermediate result didn't require closure under $\Sigma$ or $1 \in T$.
    In particular an $(n-2)$-type $X$ is a sheaf for $\langle T \rangle$ if and only if for all $Y \in T$ we have 
    \[ X \to (Y^{\ast n} \to X)\] 
    is an equivalence, where $T$ is now allowed to be any collection of types.
\end{remark}

This is a ``sheaf condition'' since we can expand $A^{\ast n}$ out to a colimit, corresponding to the covering condition for Grothendieck topologies. 
Suppose $X$ is a set, and $f : A \to B$ is a $T$-cover.
Since we present  $A \ast_B A$ as a colimit, by the universal property we have that 
\[\begin{tikzcd}
  X^{A \ast_B A}  \ar[dr, phantom, very near start, "\lrcorner"]\ar[r] \ar[d] & X^{A} \ar[d]  \\ 
  X^{A} \ar[r] & X^{A \times_B A} 
\end{tikzcd}\]
is a pullback diagram.
Exploiting the fact that $X$ is a set we can simplify this colimit, to say that 
$X^{A \ast_B A} \to X^{A} \rightrightarrows X^{A \times_B A}$ is an equaliser.
So we see that $X$ is a $T$-sheaf if and only if for any $T$-cover, $f : A \to B$ the natural map 
\[ X^B \to \lim(X^A \rightrightarrows X^{A \times_B A})\]
is an equivalence.
This is precisely a rephrasing of the sheaf condition of \cref{def:sheaf} where we have collected all of the maps in the covering into one map, by taking a disjoint union of their domains.

We also naturally obtain higher sheaf conditions. 
For example, we can present $A \ast_B A \ast_B A$ as the colimit of the diagram
\[\begin{tikzcd}
	&& A \\
	& {A \times_B A} & A \\
	&& A \\
	{A \times_B A \times_B A} & {A \times_B A} \\
	&& A \\
	& {A \times_B A} & A \\
	&& A
	\arrow[from=2-2, to=1-3]
	\arrow[from=2-2, to=2-3]
	\arrow[from=4-1, to=2-2]
	\arrow[from=4-1, to=4-2]
	\arrow[from=4-1, to=6-2]
	\arrow[from=4-2, to=3-3]
	\arrow[from=4-2, to=5-3]
	\arrow[from=6-2, to=6-3]
	\arrow[from=6-2, to=7-3]
\end{tikzcd}\]
Equivalently we can view this colimit as a partial \v{C}ech nerve of $A \to B$, in that $A \ast_B A \ast_B A$ is the colimit of 
\[\begin{tikzcd}
	{A \times_B A \times_B A} & {A\times_B A} & A
	\arrow[from=1-1, to=1-2]
	\arrow[shift right=4, from=1-1, to=1-2]
	\arrow[shift left=4, from=1-1, to=1-2]
	\arrow[shift left=2, from=1-2, to=1-1]
	\arrow[shift right=2, from=1-2, to=1-1]
	\arrow[shift left=3, from=1-2, to=1-3]
	\arrow[shift right=3, from=1-2, to=1-3]
	\arrow[from=1-3, to=1-2]
\end{tikzcd}\]
This then corresponds to the usual notion of stack condition for a stack of groupoids.
This equivalence can be increased to arbitrary length for any external natural number, but actually checking 
the coherence conditions to construct the isomorphism between the colimit of this diagram and the iterated join
 is not known to be possible internally to type theory for arbitrary $A$ and $B$.

\section{Projective Modalities}

In presheaf categories, representable objects can classified internally up to retracts as the ``tiny objects'', that is, 
those presheaves $F$ such that $\hom(F, \_)$ preserves colimits.
More specifically given a category $\mathcal{C}$, the full subcategory $\bar \cC \subseteq \mathrm{Psh}(\cC)$ of tiny objects satisfies $\mathrm{Psh}(\cC) \simeq \mathrm{Psh}(\bar \cC)$.
Tiny objects are difficult to pick out internally, since they are defined by an external right adjoint.
Recent work in modal type theory~\cite{gratzer2021mtt} and more specifically on adding tiny objects to type theory~\cite{riley2024tiny} has allowed this,
althought it requires fundamentally altering the type theory.
However, we can relax the condition somewhat, by noting that being tiny is a stronger form of being projective, a notion which can be internalised well.

\begin{definition}[Projective]
    A type $X$ is called \textbf{projective} if for any type family $B : X \to \UU$ such that $\prod_{x : X} \Trunc{B(x)}$ we have $\Trunc{\prod_{x : X} B(x)}$.
\end{definition}

A type $X$ being projective tells us type families over $X$ satisfy the axiom of choice. It can also be described in a fibred manner:
being projective is equivalent to saying surjective maps into $X$ split.
That is for all $f : Y \to X$ surjective, there is merely $g : X \to Y$ such that $f \circ g = \id$.

\begin{lemma}
\label{lem:proj_sigma}
    Projective objects are closed under $\Sigma$.
\end{lemma}
\begin{proof}
    Let $X$ be projective and $Y : X \to \UU$ be a family of projective types.
    Then consider a type family $Z : \sum_{x : X} Y(x) \to \UU$ with a term of type:
    \[ \prod_{z : \sum_{x : X} Y(x)} \Trunc{Z(z)} \]
    This type is equivalent to:
    \[ \prod_{x : X} \prod_{y : Y(x)} \Trunc{Z(x,y)} \]
    Since $Y(x)$ is projective and $X$ is projective we have a term of type:
    \[  \Trunc*{ \prod_{x : X} \prod_{y : Y(x)} Z(x, y)} \]
    By undoing the first equivalence, we obtain the necessary term.
\end{proof}

\begin{lemma}
    Projective objects are closed under sums.
\end{lemma}
\begin{proof}
    Let $X$, $Y$ be projective. 
    Then suppose we have a type family $Z : X + Y \to \UU$ so that $(z : X + Y) \to \Trunc*{Z(z)}$.
    Then we have $\prod_{x : X} \Trunc*{Z(\inl x)}$ and $\prod_{y : Y} \Trunc*{Z(\inr y)}$, and so it follows we have terms of type $\Trunc{\prod_{x : X} Z(\inl x)}$ and $\Trunc{\prod_{x : X} Z(\inl x)}$ by projectivity of $X$ and $Y$.
    Using the universal property of propositional truncation we can then construct a term $\Trunc{\prod_{z : X + Y} Z(z)}$ and we are done.
\end{proof}

\begin{example}
    The unit type $1$ is projective since the types $\prod_{x : 1} \Trunc{X(x)}$ and $\Trunc{\prod_{x : 1} X(x)}$ are both equivalent to $\Trunc{X(*)}$.
    Thus any of the finite types given by $1 + \cdots + 1$ is projective. Furthermore since being projective is a proposition, all finite types are projective.
\end{example}

\begin{definition}
    A \textbf{projective presentation} of a modality is a presentation $T$ in which all $X \in T$ are projective.
\end{definition}

The previous lemma also implies that given a collection $T$ of projective types, that $\langle T \rangle$ is a projective presentation.
For projective presentations, we can describe sheafification of a proposition explicitly.

\begin{lemma}
    \label{lem:local_props}
    For a projective presentation $T$, and any proposition $P$ we have:
    \[ \bigcirc_T P = \exists_{A \in T} P^A\]
\end{lemma}
\begin{proof}
    Modalities preserve propositions, so we can use propositional induction to define a map:
    \[ \exists_{A \in T} P^A \to \OO_T P\]
    Note for $A \in T$ we have an isomorphism $\OO_T P ^ {\Trunc A} = \OO_T P$, 
    so given $f : A \to P$ we get $| f | : \Trunc A  \to \OO_T P$ by $|f|(|a|) = \eta_P(f(a))$, which gives a point of $\OO_T P$

    For the other direction we first show that $\exists_{A \in T} P^A$ is a $T$-sheaf. Let  $X \in T$. Then we have:
    \begin{align*}
        \Trunc X  \to \exists_{A \in T} P^A 
        &= X \to \exists_{A \in T} P^A \\
        &= \Trunc*{ X \to \sum_{A \in T} P^A } \quad \text{ By projectivity of } X. \\
        &= \Trunc*{ \sum_{g : X \to T} \prod_{x : X} P^{g(x)} }
    \end{align*}
    Given such a $g$ we obtain a map $\sum_{x: X} g(x) \to P$, and since $T$ is closed under $\sum$, we have $\sum_{x: X} g(x)$ is in $T$. 
    So we have $\Trunc X  \to \exists_{A \in T} P^A = \exists_{A \in T}P^A$ for any $X \in T$ and thus the type is local.
    
    Now we define a map $\OO_T P \to \exists_{A \in T} P^A $.
    Since we have established the right-hand side is local, it's enough to define a map $P \to \exists_{A \in T} P^A$, which we can write down explicitly as $p \mapsto | (1, \mathrm{const}_p) |$.
\end{proof}

As a corollary we can deduce ordinary results from the Kripke-Joyal semantics, for example:
\begin{corollary}
    Let $T$ be a projective presentation of a modality.
    Then given a type $X$ and a predicate $\phi : X \to \Prop$ we have:
    \[ \OO_T \exists_{x :X} P(x) = \exists_{A \in T} \exists_{x : A \to X} \prod_{a : A} P(x(a)) \]
\end{corollary}
\begin{proof}
    Follows quickly from lemma~\ref*{lem:local_props}, and projectivity.
\end{proof}

Using projectivity, we can show a local version of the axiom of choice holds true in the subuniverse.
This principle has come to be known as ``local choice'', and has had striking applications in synthetic algebraic geometry~\cite{cherubini2023foundation}.
This more general principle is also known externally, for $1$-topoi. 
The benefit of this internal version is the ability to apply it to higher categories.

\begin{definition}
    A map $f : X \to Y$ is \textbf{$T$-surjective} if for all $y : Y$ we have a term of type $\OO_T \Trunc{\fib_f(y)} $.
    This is the externalisation of the condition of being surjective in the subtopos.
\end{definition}

\begin{lemma}[$T$-local partial choice]
\label{lem:local_choice}
    Consider a diagram
    \[\begin{tikzcd}
        & X \ar[d, "f"] \\
        D \ar[r, "g"] & C 
    \end{tikzcd}\]
    where $D$ is projective and $f$ is $T$-surjective.
    Then there exists a type $Z$, with a map $Z \to X$ and a $T$-cover $Z \to D$ such that
    \[\begin{tikzcd}
       & & X \ar[d, "f"] \\
       Z \ar[r] \ar[urr] & D \ar[r, "g"] & C 
    \end{tikzcd}\]
    commutes.
\end{lemma}
\begin{proof}
    By $T$-surjectivity and Lemma~\ref{lem:local_props} we have a term of type
    \[  \prod_{d : D} \OO_T \Trunc{\fib_f(g(d))} = \prod_{d : D} \Trunc*{\sum_{Y \in T} \Trunc{ \fib_f(g(d))}^Y } \]
    Since $D$ is projective, and we are proving a proposition, we obtain a term
    \[ \phi : \prod_{d : D} \sum_{Y \in T} \Trunc{ \fib_f(g(d))}^Y  \]
    Now set $Z :\equiv \sum_{d : D} \pi_1(\phi(d))$.
    Projection down to $d$ then has fibres $\pi_1(\phi(d))$, which are in $T$ by assumption, making $Z \to D$ a $T$-cover.
    We define a map $Z \to X$ by noting we have a term $\prod_{(d, y) : Z} \Trunc{ \fib_f(g(d)) }$ given by $(d, y) \mapsto \pi_2(\phi(d))(y)$.  
    Since projective objects are closed under $\Sigma$ we note that $Z$ is projective, and we get merely $\prod_{(d, y) : Z} \fib_f(g(d))$, which then merely gives a map $Z \to X$ making the diagram commute.
\end{proof}

Note we can rephrase this more similarly to the axiom of choice, in an indexed fashion.
$T$-local choice implies that for all $D$ projective and all $B : D \to U$ such that $\OO_T \Trunc{B(d)}$ for all $d : D$, there is merely a $T$-cover $f : Z \to D$
such that $B(f(z))$ for all $z : Z$.
Thus we are allowed to make choices on covers, but these choices might not glue to form a whole section.

\section{Cohomology}

Cohomology detects failures of gluing problems. 
Every topos comes with a cohomology theory, related to gluing objects defined on the covers of the topology.
Normally the definition of cohomology is non-constructive, requiring the axiom of choice to ensure the given category has enough injective abelian groups.
However in higher topoi, this problem is averted.
Cohomology functors are all representable, and the representing objects are easily described.
We use this to our advantage to describe cohomology theories for given subuniverses.

\begin{definition}[\cite{licata2014EMspace}]
  \label{def:EM}
    Given a group $G$, and $n : \bN$, the $n$\textsuperscript{th} Eilenberg-MacLane space is denoted $K(G, n)$.
    It is a pointed $n-1$-connected, $n$-truncated type satisfying $\Omega K(G, n) = K(G, n-1)$ with $K(G, 0) = G$.
\end{definition}

\begin{definition}
    Given a type $X$, and a group $G$, we define its $n$\textsuperscript{th} \textbf{cohomology group} as
    \[ H^n(X, G) := \Trunc{X \to K(G, n)}_0\]
\end{definition}

There is a close relationship between the axiom of choice and cohomology. 
In fact one way to understand cohomology in a ``universal'' way is as an obstacle to the axiom of choice holding for a certain type.
For example, it is easily seen that any projective type has zero cohomology:

\begin{proposition}
    Let $X$ be a projective type.
    Then for all $n > 0$ we have $H^n(X, G) = 0$.
\end{proposition}
\begin{proof}
    We need to show that for all $\chi : X \to K(G, n)$ that $\Trunc{\chi = \mathrm{const}}$, where $\mathrm{const}$ is the constant map at the point of $K(G, n)$.
    But note, since $n \geq 1$, we have that $K(G, n)$ is $0$-connected.
    Hence we have for all $x : X$ that $\Trunc{\chi(x) = \mathrm{pt}}$.
    But $X$ is projective so we deduce that $\chi$ is merely the constant map and we are done.
\end{proof}

Since, given a projective presentation, local choice holds, we may wonder whether this can be used to compute cohomology, and compare cohomology in a subuniverse.
For this we will need to consider the localisation $\OO K(G, n)$, for $G$ a group sheaf.
For any lex modality this is precisely the Eilenberg-Maclane space in the subuniverse.

We briefly check that this does correspond to the correct subobject.
\begin{definition}
  Given a lex modality $\OO$, the $k$-truncation of a sheaf $X$ in the subuniverse is $\OO \Trunc{X}_k$.
  We will call a $k$-truncated sheaf an $\OO$-$k$-truncated type.
  We will say a sheaf $X$ is $\OO$-$k$-connected if $\prod_{x, y : X} \OO \Trunc{x = y}_k$, analogously.
\end{definition}

\begin{lemma}
  \label{lem:modal_connectivity}
  Suppose $X$ is a $k$-connected type. Then $\OO X$ is $\OO$-$k$-connected.
\end{lemma}
\begin{proof}
  To define a map $\prod_{x, y : \OO X} \OO \Trunc{x = y}_k$ we apply modal induction and see it is enough when $x, y : X$.
  We have $X$ is $k$-connected so that $\Trunc{x = y}_k$ and hence $\Trunc{\eta(x) = \eta(y)}_k$ in $\OO X$.
  Thus we obtain a term of $\OO \Trunc{\eta(x) = \eta(y)}_k$ and we are done.
\end{proof}

\begin{proposition}
  For any group sheaf $G$, the type $\OO K(G, n)$ is an $n$\textsuperscript{th} Eilenberg-MacLane of $G$ 
  in the subuniverse $U_\OO$.
\end{proposition}
\begin{proof}
  Translating~\cref{def:sheaf} into the subuniverse, we need to show that $\OO K(G, n)$ is $\OO$-$(n-1)$-connected, $\OO$-$n$-truncated and that $\Omega \OO K(G, n) = \OO K(G, n-1)$.
  Finally we need to check $\OO K(G, 0) = G$.
  The first fact follows since any modality preserves connectivity~\ref{lem:modal_connectivity}, 
  and by definition $K(G, n)$ is $n-1$-connected.
  Similarly lex modalities preserve truncation levels~\cite[Corollary 3.9]{spitters2020modalities}, it follows that $\OO K(G, n)$ is also $\OO$-$n$-truncated.

  Finally lex modalities satisfy $\OO (x = y) \simeq \eta(x) = \eta(y)$~\cite[Theorem 3.1.ix]{spitters2020modalities}.
  Hence we can deduce that for $\eta(\mathrm{pt}) : \OO K(G, n)$ we have:
   \[ \eta(\mathrm{pt}) = \eta(\mathrm{pt}) = \OO(\mathrm{pt} = \mathrm{pt}) = \OO K(G, n-1) \]
  Finally we have $\OO K(G, 0) = \OO G = G$ since $G$ is a sheaf.
\end{proof}

Hence we can use the Eilenberg MacLane space in the subuniverse to define cohomology..
\begin{definition}
    Given a lex modality $\OO$ we can define its cohomology theory. 
    For any sheaf $X$ and any group sheaf $G$ we define
    \[ H^n_\OO(X, G) := \OO \Trunc{ X \to \OO K(G, n) }_0 \]
\end{definition}

\begin{proposition}
  For any $H^n_\OO(X, G)$ is the externalisation of the internal cohomology of $X$ 
\end{proposition}
\begin{proof}
  By the previous proposition, $\OO K(G, n)$ is the Eilenberg MacLane space of $G$ in the subuniverse.
  Similarly $\OO\Trunc{\_}_0$ is the set truncation.
  Thus this is exactly the translation of the definition of cohomology into the subuniverse.
\end{proof}

\begin{remark}
    There is a canonical map $\OO H^n(X, G) \to H^n_\OO(X, G)$ defined as follows.
    Since $H^n_\OO(X, G)$ is a sheaf it is enough to define a map from $H^n(X, G)$.
    Similarly since $H^n_\OO(X, G)$ is a set, it is enough to define the untruncated map.
    Given a term $\chi : X \to K(G, n)$ we obtain a term $\eta \circ \chi : X \to \OO K(G, n)$,
    which gives a term of $H^n_\OO(X, G)$.
\end{remark}

A key use for the local choice principle shown in~\cref{lem:local_choice} is computing cohomology classes
in the subuniverse.
In particular, a cohomology class $\chi : X \to \OO K(G, 1)$ satisfies 
\[ \prod_{x : X} \OO  \Trunc{\chi(x) = \mathrm{pt}}\] and if $X$ is projective then we can find a cover $\phi : Z \to X$
such that \[\Trunc*{\prod_{z : Z} \chi(\phi(z)) = \mathrm{pt}}  \]
We can then classify how these glue to global sections on $X$ to determine the cohomology classes.

We now determine a class of abelian groups so that cohomology agrees in the subuniverse, up to sheafification.
Note that for any type $X$ and any abelian group $A$, we have a sequence
\[  A^X \to A^{X \times X} \to A^{X \times X \times X }  \]
where the two maps are 
\[  d : A^X \to A^{X \times X} \]
\[ d(f)(x, x') := f(x) - f(x') \]
  
\[ d : A^{X \times X} \to A^{X \times X \times X} \]
\[ d(f)(x, x', x'') := f(x, x') - f(x, x'') + f(x',x'') \]

Given a presentation $T$ 
we will say $A$ satisfies \textbf{descent for $T$} if for all $X \in T$ the above sequence is exact.

\begin{lemma}
  \label{lem:desc_exact_cover}
    Suppose $A$ satisfies descent for a presentation $T$.
    Then for all $T$-covers $\phi : Z \to X$ there is an exact sequence
    \[ A^{Z} \to A^{Z \times_X Z} \to A^{Z \times_X Z \times_X Z} \]
\end{lemma}
\begin{proof}
    For each $x : X$ write $Z_x$ for the fiber of $\phi$ at $x$. Then for each $x$ we have an exact sequence
    \[  A^{Z_x} \to A^{Z_x \times Z_x} \to A^{Z_x \times Z_x \times Z_x}  \]
    by assumption.
    Since product preserves exactness we have 
    \[ \prod_{x : X} A^{Z_x} \to \prod_{x : X} A^{Z_x \times Z_x} \to \prod_{x : X} A^{Z_x \times Z_x \times Z_x} \]
    is exact.
    But $\prod_{x : X} A^{Z_x} = A^{(x : X) \times Z_x} = A^Z$
    and similarly for the other two terms.
    The other products simplify similarly.
\end{proof}

\begin{lemma}[Gluing]
  \label{lem:gluing}
    Let $X$ be a type, $B : X \to \UU$ be a pointwise merely inhabited type family,
    and $C : X \to \Set_\UU$ be a family of $0$-types.
    Suppose we have 
    \[ s : \prod_{x : X} \prod_{b : B(x)} C(x) \]
    such that 
    \[ H : \prod_{x : X} \prod_{b_1 b_2 : B(x)} s_x(b_1) = s_x(b_2) \]
    Then there exists a map 
    \[ s' : \prod_{x: X} C(x) \]
    extending $s$. 
\end{lemma}
\begin{proof}
  Let $x : X$. 
  Then we have $ \Trunc{B(x)} $ since $B$ is inhabitted.
  So it suffices to find a map $\Trunc{B(x)} \to C(x)$.
  But since $C(x)$ is a set, maps $\Trunc{B(x)} \to C(x)$ correspond to constant maps $B(x) \to C(x)$.
  We have precisely this data by $s_x$, by assumption.
\end{proof}

\begin{theorem}
  \label{th:zero_cohom}
    Let $T$ be a projective presentation and $A$ an abelian group satisfying descent for $T$.
    Then for all projective $X$, we have $H^1_\OO(X, A) = 0$, where $0$ is the trivial abelian group.
\end{theorem}
\begin{proof}
    Let $\chi : X \to \OO K(A, 1)$ be a cohomology class.
    As before we can find a $T$-cover $\phi : Z \to X$ and a term 
    \[ s : \Trunc*{\prod_{z : Z} \chi(\phi(z)) = \mathrm{pt}} \] 
    Now we (merely) define a map $t : Z \times_X Z \to A$ by $s(z)^{-1} \cdot s(z') : \mathrm{pt} = \mathrm{pt}$.
    Note that $(\pt = \pt) \simeq A$ by definition of $K(A,1)$, exploiting lexness of $\OO$.
    By definition we can see that $t$ maps to $0$ in $Z \times_X Z \times_X Z \to A$ by $d$.
    Hence, by~\cref{lem:desc_exact_cover}, there is merely $u : Z \to A$ such that $d(u) =t $.
    Then we can merely define $s' : \prod_{z : Z} \chi(\phi(z)) = \mathrm{pt}$
    by $s' = s \cdot u^{-1}$.
    This satisfies for all $(z_1, z_2) : Z \times_X Z$ that $s'(z_1) = s'(z_2)$.
    Hence, noting that the equality types of $\OO K(A, 1)$ are all sets, by definition of $K(A, 1)$ and lexness of $\OO$,
    we may apply~\cref{lem:gluing} and glue $s'$ to a section 
    \[ s'' : \prod_{x : X} \chi(x) = \pt \]
    We then deduce that $\chi$ is merely the constant map, and hence is trivial in $\Trunc{X \to \OO K(A, 1)}_0$.
\end{proof}

\begin{example}
  Every abelian group $A$ satisfies descent for the trivial presentation $\langle 1 \rangle$,
  since the required sequence is just formed of identity maps.
  Further examples for more interesting presentations will be showed later.
\end{example}

\begin{remark}
  Using \v{C}ech cohomology results~\cite{blechschmidt2024cech}, this result should generalise to types built out of projectives in a nice way.
\end{remark}

\section{Applications}

We give a generic example application, which we then specialise to synthetic algebraic geometry, and directed type theory.
Type theory interprets into topoi, and thus we are able to axiomatise our theory to have a model in a specific topos.
A very common example is the classifying topos of an algebraic theory.
For example if we would like to study algebraic geometry, we might consider modelling into functors from finitely presented commutative rings to $\mathrm{Set}$, 
that is, the classifying topos of the theory of rings.
Or for directed type theory, we might consider modelling into presheaves on cubes or lattices, which also classify the corresponding algebraic theories.

In recent work of Ingo Blechschmidt~\cite{blechschmidt20XXqcoh} a strong generic axiom was shown to hold in these classifying categories
, originally called synthetic quasi-coherence in the setting of algebraic geometry,
 now more commonly referred to as duality in the more general setting.
We will add axioms and enforce duality, to obtain a type theory which can be interpreted into a classifying category.
We will then build projective modalities in these theories, and establish some useful consequences.

Consider an algebraic theory $\bT$.
We add the following axioms to form a presheaf model in type theory of the theory $\bT$.

\begin{axiom}
    We have a model $U_\bT$ of $\bT$ over signature $\Sigma$.
\end{axiom}

From this we are able to construct models over $U_\bT$. We call a model $A$ of $\bT$ with a model homomorphism $U_\bT \to A$ a $U_\bT$-algebra. 
A homomorphism of $U_\bT$ algebras is given by a map of models $A \to B$ so that the following diagram commutes.
\[\begin{tikzcd}
    A \ar[rr]  & & B \\
    & U_\bT \ar[ul] \ar[ur]
\end{tikzcd}\]
$U_\bT$-algebras can be described by a new algebraic theory, axiomatising the map from $U_\bT$.

For any algebraic theory $T$ we can, given a set $X$ of variables and a set $R$ of atomic formulas, construct an algebra 
\[ \bT[X]/R := \mathrm{Terms}(\Sigma + X) / \sim  \]
presented by $X$ and $R$.
$\Sigma + X$ is a new signature with a constant symbol for each $x : X$ and  
two terms $t_1 \sim t_2$ are identified if $\bT + R \vdash t_1= t_2$.

In particular we can apply this to the theory of $U_\bT$-algebras to get $U_\bT$-algebras presented by generators and relations: $U_\bT[X]/R$. 

We call a presentation \textbf{finite} if $X$ and $R$ are finite.
~\footnote{Scholars of finiteness will recognise there are different notions of finite in constructive mathematics. Specifically we mean $X$ is Bishop-finite and $R$ is Kuratowski-finite - since only the image of $R$ is required in the construction.}
This comes with a natural algebra map $U_\bT \to U_\bT  [X] / R $.
We say a generic $U_\bT$ algebra is finitely presented if it is merely equal to a one of these for some $X$, $R$.

\begin{definition}
    For a finitely presented $U_\bT$-algebra $A$ we define its \textbf{spectrum} to be
    \[ \Spec(A) :\equiv \Hom_{U_\bT}(A, U_\bT) \]
    We call a type \textbf{affine} if it is merely equal to $\Spec(A)$ for some finitely presented algebra $A$.
\end{definition}

Note that $\Spec$ is functorial: by composition, it sends a homomorphism $f : A \to B$ to a map $\Spec(B) \to \Spec(A)$
We also have a natural functor in the reverse direction, taking any affine type $X$ to the free $U_\bT$ algebra on it, by ${U_\bT}^X$.
Additionally, for any finitely presented $U_\bT$-algebra $A$ there is a natural map
\[ A \to {U_\bT}^{\Spec(A)} \]
\[ a \mapsto (\phi \mapsto \phi(a))\]
We add an axiom to our type theory forcing that $\Spec$ is an equivalence of categories, which is true in models as in \cite{blechschmidt20XXqcoh}.

\begin{axiom}[Blechschmidt duality]
    For any finitely presented $U_\bT$-algebra $A$ the natural map $A \to {U_\bT}^{\Spec(A)}$ is an equivalence.
\end{axiom}

\begin{lemma}
    $\Spec$ is an equivalence of between the categories of finitely presented models, and the opposite category of affine types and functions between them.
\end{lemma}
\begin{proof}
    We already have one composition gives the identity by duality.
    Now we show $\Spec({U_\bT}^X) = X$ for $X$ an affine type. 
    Write $X = \Spec(A)$ then $\Spec({U_\bT}^{\Spec(A)}) = \Spec(A)$. 
    It isn't hard to see this is a natural isomorphism.
\end{proof}

Note since, by standard category theory, the functor $\Hom(\_, U_\bT)$ is limit preserving, we have
\[ \Spec(\colim_{i : I}  A_i) = \lim_{i : I} \Spec(A_i)\]

\begin{axiom}
    For any finitely presented $U_\bT$-algebra $A$ we have $\Spec(A)$ is projective.
\end{axiom}

Note in the classifying topos of an algebraic theory, the tiny objects are precisely given by finite sums of affines, as the category is idempotent complete.
Thus we will consider presentations given by these objects. 

\begin{definition}
    We call a presentation $T$ \textbf{subcanonical} if every affine type is a $T$-sheaf.
\end{definition}

\begin{lemma}
  \label{lem:subcan}
    A presentation $T$ is subcanonical iff $U_\bT$ is a $T$-sheaf.
\end{lemma}
\begin{proof}
    ($\Rightarrow$) Suppose $T$ is subcanonical. 
    Then $U_\bT = \Spec(U_\bT[x])$ is a sheaf by definition.

    ($\Leftarrow$) Suppose $U_\bT$ is a $T$-sheaf.
    Then for any f.p algebra $A$ we have 
    $\Spec(A) := (f : A \to U_\bT) \times H(f)$
    where $H(f)$ are the equations in $R$. 
    Since $U_\bT$ is a sheaf so is $A \to U_\bT$ and the equations are also all sheaves.
\end{proof}

Since $U_\bT$ is a set by assumption, we can always apply the set sheaf condition after \cref{cor:sheaf_cond}.
So $T$ is subcanonical if and only if for all $A \in T$ we have 
\[ U_\bT \to \lim( (U_\bT)^A \to (U_\bT)^{A \times A} )\]
is an equivalence.
The topologies we consider will have $T$ be generated by finite sums of affines.
In particular this means we can exploit duality with $U_\bT$.
Set $A = \Spec(X_1) + \cdots + \Spec(X_n)$.
Note that as the category of models of any algebraic theory is cocomplete, and we write the pushout of $A \leftarrow C \rightarrow B$ as $A \otimes_C B$.
Then duality, combined with limit preservation of $\Spec$, tells us that the $U_\bT$ is a sheaf if and only if 
\[ U_\bT \to \lim \left ( \prod_{i} X_i \rightrightarrows \prod_{i, j} X_i \otimes_{U_\bT} X_j \right ) \]
is an equivalence for all $i, j$.
This is now a purely algebraic question, and is the sort of gluing condition algebraic geometers are very familiar with.

Finally, we consider an additional axiom which will allow for our presentations to genuinely be closed under $\Sigma$.
\begin{axiom}
    For all finitely presented $U_\bT$-algebras $A$, the diagonal map $\mathbb{N} \to \bN^{\Spec(A)}$ is an equivalence.
    That is every map $\Spec(A) \to \mathbb{N}$ is constant.
\end{axiom}

This is justified by a simple external calculation:
\begin{lemma}
    Let $\cC$ be a category with an initial object $0$.
    Let $X : \Set$ and abusing notation let $X$ also represent the constant presheaf with value $X$.
    Then for any $ c : \cC$ the diagonal map $X \to X^{\yo_c}$ is an equivalence.
\end{lemma}
\begin{proof}
    Let $d : \cC$. Then unwinding the definition $X^{\yo_c}(d) = \hom(\yo_c \times \yo_d, X)$
    So let $\eta : \yo_c \times \yo_d \to X$.
    Then for all $z : \cC$ we have the unique map $0 \to z$ in $\cC$ making the following diagram commute:
    \[\begin{tikzcd}
        \hom(z, c) \times \hom(z, d) \ar[r, "\eta_z"] \ar[d]  &= X \\
        \hom(0, c) \times \hom(0, d) \ar[ur, "\eta_0", swap]
    \end{tikzcd}\]
    But since $\hom(0, c) \times \hom(0, d) = \star$ we deduce that each component of $\eta$ factors through $\star$ and hence is constant.
    It is straightforward to show that this defines an isomorphism of $X^{\yo_c} \simeq X$ sending each $\eta$ to its constant value, hence inverting the diagonal map.
\end{proof}

This justifies the axiom since in our type theory $\mathbb{N}$ corresponds to the constant presheaf at the natural numbers,
and since our representables are the opposite category of finitely presented algebras, the initial object is given by the zero algebra.

\begin{lemma}
    Given $A$, $B$ finitely presented $U_\bT$ algebras, we have a natural equivalence
    \[ A \otimes_{U_{\bT}} B = \Spec(A) \to B\] 
\end{lemma}
\begin{proof}
    This is a simple calculation using duality, since the coproduct of finitely presented algebras is finitely presented.
    \begin{align*}
        A \otimes_{U_{\bT}} B &= \Spec(A \otimes_{U_{\bT}} B) \to U_\bT \\
                            &= \Spec(A) \times \Spec(B) \to U_\bT \\
                            &= \Spec(A) \to (\Spec(B) \to U_\bT) \\
                            &= \Spec(A) \to B
    \end{align*}
\end{proof}

\begin{lemma}
    Affines are closed under $\Sigma$.
\end{lemma}
\begin{proof}
    Let $A$ be an f.p. $U_\bT$ algebra and $B$ be a family of $U_\bT$ algebras depending on $\Spec(A)$.
    Then note we have for all $x : \Spec(A)$ that $B(x)$ is merely of the form $U_\bT [ X_1 , \ldots, X_n ] / (s_1 = t_1, \ldots, s_m = t_m)$ for some $n, m :\bN$ and terms $s_i, t_i : U_\bT \{ X_1 , \ldots, X_n\}$.
    By projectivity of $\Spec(A)$ we are able to explicitly obtain functions $n, m : \Spec(A) \to \bN$ giving the length of each $B(x)$ which must be constant by our axioms.
    Then we are also able to explicitly obtain functions $s_1, \ldots, s_m, t_1, \ldots, t_m : \Spec(A) \to U_\bT [ X_1, \ldots, X_n ]$ giving the equations.
    By the previous lemma these are equivalently terms of $A \otimes_{U_\bT} U_\bT [ X_1, \ldots, X_n ] = A [ X_1, \ldots, X_n ]$.
    We now claim that 
    \[ \sum_{x : \Spec(A)} B(x) = A [X_1 , \ldots X_n ] / (s_1 = t_1, \ldots s_m = t_m) \]
    But this is true since the natural map
    \[ \Spec(A [X_1, \ldots, X_n] / (s_1 = t_1, \ldots, s_m = t_m)) \to \Spec(A) \]
    has fiber at $x : \Spec(A)$ given by $r(x)$, using that $\Spec$ preserves limits.
\end{proof}

\begin{remark}
    A simple corollary is that finite sums of affines are closed under $\Sigma$ by applying constancy of maps to $\bN$ and distributing over the sum.
\end{remark}

\subsection{Synthetic Algebraic Geometry}

We will now apply this general example to a specific case, following~\cite{moeneclaey2024sheaves}.
Let $\bT$ be the theory of commutative rings with unit. 
We write $R := U_\bT$ and consider the \textbf{Zariski topology} $\Zar$, for which we give the presentation containing types merely of the form
\[ \Spec(R_{f_1}) + \cdots + \Spec(R_{f_n}) \]
where $f_1, \ldots , f_n : R$ satisfy $1 \in (f_1, \ldots, f_n)$, and the ring $R_{f_i}$ is $R$ with the element $f_i$ inverted.

\begin{lemma}
    $\Zar$ is a projective presentation.
\end{lemma}
\begin{proof}
    The unit $1 \in \Zar$ by choosing the sequence $1$, since $\Spec(R_{1}) = \Spec(R) = 1$.
    For $\Sigma$-closure, let $X \in \Zar$ and $Y : X \to \Zar$
    Then $\Trunc{ X = \Spec(R_{f_1}) + \cdots + \Spec(R_{f_n})}$ and we have 
    \[ \prod_{x : X} \Trunc*{ \sum_{g_i : R} Y(x) = \Spec(R_{g_1}) + \cdots \Spec(R_{g_m}) } \]
    By projectivity of $X$, we can push the product inside the truncation.
    We can rewrite this commuting product and sum to get a function $\phi : X \to \mathbb{N}$ giving the length of the sequence, and terms $g : (x : X) \to R^{\phi(x)}$ for the quotients.
    But every map from a spectrum to the naturals is constant, so we merely obtain $m_1, \ldots, m_n$ such that for $x : \Spec(R_{f_i})$ we have $\phi(x) = m_i$.
    Then we calculate
    \begin{align*}
        (x : \Spec(R_{f_i}) + &\cdots + \Spec(R_{f_n})) \times (\Spec(R_{g_{i,1}(x)}) + \cdots \Spec(R_{g_{i,\phi(x)}(x)})) 
        \\ &= \sum_{i = 1}^n \sum_{x : \Spec(R_{f_i})} \Spec(R_{g_{i,1}(x)}) + \cdots \Spec(R_{g_{i,\phi(x)}(x)})
        \\ &= \sum_{i = 1}^n \sum_{x : \Spec(R_{f_i})} \Spec(R_{g_{i,1}(x)}) + \cdots \Spec(R_{g_{i, m_i}(x)})
        \\ &= \sum_{i = 1}^n \sum_{j = 1}^{m_i} \sum_{x : \Spec(R_{f_i})} \Spec(R_{g_{ij}(x)}) 
        \\ &= \sum_{i = 1}^n \sum_{j = 1}^{m_i} \Spec( {R_{f_i}}_{g_{ij}} )
    \end{align*}

    Now this proof resolves down to the standard algebraic proof that the Zariski site forms a Grothendieck topology.
    Here $g_{ij} : \Spec(R_{f_i}) \to R$ is interpreted as an element of $R_{f_i}$ by duality.
    By assumption these satisfy that $1 \in (g_{i1}, \ldots, g_{i m_i})$.
    Writing $g_{ij} = r_{ij} / f_i^{k_{ij}}$ we can see that localising at $g_{ij}$ is the same as localising at $r_{ij}$.
    So we need only show that $1 \in (f_i r_{i j})$ instead.
    As an ideal in  $R_{f_i}$, we have  $1 \in (r_{i1}, \ldots, r_{i m_i})$ clearly. 
    Multiplying out denominators we see there is a natural number $k_i$ such that $f_i^{k_i} \in (f_i r_{i1}, \ldots , f_i r_{i m_i})$ as an ideal in $R$.
    Since $1 \in (f_1 , \ldots, f_n)$ we deduce that $1 \in (f_1^{k_1}, \cdots, f_n^{k_n})$ and hence are in the ideal generated by $(f_i r_{i j})$. 
\end{proof}

We show this presentation is subcanonical: by \cref{lem:subcan}, we only need show that $R$ is a sheaf for the presentation.
Let $A = \Spec(R_{f_1}) + \cdots + \Spec(R_{f_n})$ be in $\Zar$. 
Note a standard fact about localisation is that $R_{f_i} \otimes_R R_{f_j} = R_{f_i f_j}$.
Hence the sheaf condition states that $R$ is a sheaf if and only if
\[ R \to \lim \left ( \prod_{i} R_{f_i} \rightrightarrows \prod_{i,j} R_{f_i f_j} \right ) \]
is an equivalence.
Writing $[\_]_{ij}$ for the natural maps $R_{f_i} \to R_{f_i f_j}$ and similarly $[\_]_{i}$ for $R \to R_{f_i}$ we see this as the following algebraic statement:
For any $(x_1, \ldots, x_n) : R_{f_1} \times \cdots \times R_{f_n}$ such that $[x_i]_{ij} = [x_j]_{ij}$ for all $i, j$, there is a unique $x : R$ such that $[x]_i = x_i$.
This might seem like a somewhat contrived property if you've not done any algebraic geometry before, but this is in fact incredibly natural, and is proved (in some form) at the start of any algebraic geometry course. 
For a constructive proof of results similar to these, see the local-global principles in~\cite{Lombardi_2015}.
Thus $R$ is a sheaf and $\mathrm{Zar}$ is subcanonical.

Now we consider the subtopos carved out by the Zariski topology. 
In this subuniverse, Blechschmidt duality still holds, and since $\Zar$ is a projective presentation, the local choice principle of \cref{lem:local_choice} holds.
In this specific example, we can characterise the $\Zar$-covers of any finitely presented $R$-algebra.

\begin{lemma}
    Let $A$ be a finitely presented $R$-algebra. Then any Zariski cover of $\Spec(A)$ is merely of the form $\Spec(A_{f_1}) + \cdots + \Spec(A_{f_n})$ where $f_1, \ldots, f_n : A$ satisfy $f_1 + \cdots + f_n = 1$.
\end{lemma}
\begin{proof}
    Let $Z \to \Spec(A)$ be a Zariski cover.
    This can dually be represented by its fibre, $r : \Spec(A) \to \Zar$.
    We know that $\prod_{x : \Spec(A)} \Trunc{ r(x) = \Spec(R_{f_1}) + \cdots + \Spec(R_{f_n})}$ by definition of the Zariski presentation.
    By projectivity of $\Spec(A)$ we can push the product inside the truncation.
    This allows us to merely obtain a map $\Spec(A) \to \bN$ giving the length of these sequences, which must be constant by our axioms.
    Writing $n$ for the value of the length, we can view each $f_i$ as a function $\Spec(A) \to R$ or as an element of $A$.
    Thus we obtain an equivalence
    \[ \sum_{x : \Spec(A)} r(x) = \sum_{x : \Spec(A)} \Spec(R_{f_1(x)}) + \cdots \Spec(R_{f_n(x)}) = \Spec(A_{f_1}) + \cdots + \Spec(A_{f_n}) \] 
\end{proof}

This recovers the usual statement of Zariski local choice as stated in the axioms of synthetic algebraic geometry.
\begin{lemma}
    For any finitely presented $R$-algebra $A$ and any Zariski merely inhabited type family $B : \Spec(A) \to \UU$, 
    there are $f_1, \ldots, f_n : A$ such that $f_1 + \cdots + f_n = 1$ and choices of sections $s_i : (x : \Spec(A_{f_i})) \to B(x)$. 
\end{lemma}

There are several other presentations of interest in the case of synthetic algebraic geometry, containing the Zariski presentation.
In particular we have the \'etale topology, generated by the types in the Zariski topology and $\Spec(R[x]/g)$ for any unramifiable polynomial $g$.
Similarly we can take the fppf topology, whose terms are generated by $\Spec(A)$ for $A$ faithfully flat over $R$.

\subsubsection*{Cohomology}

We now consider the cohomology of the subuniverses in the setting of synthetic algebraic geometry. 

\begin{definition}
    We call an $R$-module $M$ \textbf{weakly quasicoherent} if for all $f : R$ we have $M^{\Spec(R_f)} = M \otimes_R R_f = M_f$
\end{definition}

\begin{lemma}
  \label{lem:wqc_descent}
    Any weakly quasicoherent module $M$ satisfies descent for the Zariski topology.
    That is for all $X \in \Zar$ we have an exact sequence
    \[  M^X \to M^{X \times X} \to M^{X \times X \times X } \]
\end{lemma}
\begin{proof}
    Let $X = \Spec(R_{f_1}) + \cdots \Spec(R_{f_n})$.
    Then rewriting using duality and weak quasicoherence, and standard properties of algebra, the exact sequence becomes 
    \[ \prod_{i} M_{f_i} \to \prod_{i,j} M_{f_i f_j} \to \prod_{i,j,k} M_{f_i f_j f_k} \]
    This sequence is exact, as is shown in any algebra course.
\end{proof}

\begin{corollary}
    We have 
    \[ H^1_{\Zar}(\Spec(A), M) = 0\]
    for any fintiely prestented $R$-algebra $A$ and any weakly quasicoherent module $M$.
\end{corollary}
\begin{proof}
  Since weakly quasicoherent satisfy descent, by~\cref{lem:wqc_descent}, we can apply \cref{th:zero_cohom} to deduce this.
\end{proof}

We can extend this result to include both the fppf and \'etale presentation by picking a smaller class of modules. 
\begin{definition}
    Call an $R$-module $M$ \textbf{quasicoherent} if for all finitely presented $R$-algebras $A$ we have $M^{\Spec(A)} = M \otimes_R A$.
\end{definition}
Using this standard definition we can abstract the property of the zariski topology that allows for descent to occur.
The particular property is that $\prod_{i} R_{f_i}$ is \textbf{faithfully flat} whenever $1 \in (f_1 + \cdots f_n)$.

\begin{definition}
  An $R$-algebra $A$ is \textbf{faithfully flat} if tensoring with $A$ preserves short exact sequences.
\end{definition}

\begin{lemma}
  Let $A$ be a faithfully flat $R$-module. 
  Then for any $R$-module $M$ we have An exact sequence 
  \[ M \otimes_R A \to M \otimes_R A \otimes_R A \to M \otimes_R A \otimes_R A \otimes_R A \]
\end{lemma}

\begin{theorem}
  Let $T$ be a presentation which satisfies, that any  
   $X \in T$ can be writen as $X = \Spec(A_1) + \cdots + \Spec(A_n)$ where $A_1 \times \cdots \times A_n$ is a faithfully flat $R$-algebra. 
  
  Then all quasicoherent modules satisfy descent for $T$.
\end{theorem}
\begin{proof}
  Let $X = \Spec(A_1) + \cdots + \Spec(A_n) \in T$.
  Write $A = A_1 \times \cdots \times A_n$.
  Then we have 
  \begin{align*}
    M^X \to &M^{X \times X} \to M^{X \times X \times X} \\
    &=  \prod_{i} M \otimes_R A_i  \to \prod_{ij} M \otimes_R A_i \otimes_R A_j \to \prod_{ijk} M \otimes_R A_i \otimes_R \otimes_R A_j \otimes_R A_k \\
    &= M \otimes_R A \to M \otimes_R A \otimes_R A \to M \otimes_R A \otimes_R A \otimes_R A 
  \end{align*}
  which is exact by the previous lemma.
\end{proof}

\begin{corollary}
  Let $T$ be either the fppf, \'etale or zariski presentations.
  Then for any fintiely presented $R$-algebra $A$ and any abelian group $G$ we have $H^1(\Spec(A), G) = 0$  
\end{corollary}
\begin{proof}
  All three presentations satisfy the condition of the previous theorem.
\end{proof}

\subsection{Triangulated Type Theory}

The recent extension of simplicial HoTT by \citeauthor{gratzer2024directedunivalencesimplicialhomotopy}~\cite{gratzer2024directedunivalencesimplicialhomotopy} is closely related to this example.
We take the subset of triangulated type theory given by an instance of our generic example with $\bT$ as the theory of bounded distributive lattices.
We write the universal model of this theory as $\bI$. They then define the following

\begin{definition}
    A type $X$ is simplicial if it is local with respect to $(i \leq j) \vee (j \leq i)$ for all $i, j : \bI$.
\end{definition}
Note that
$(i \leq j) \vee (j \leq i) = \Trunc{\Spec(\bI/(i \leq j)) + \Spec(\bI/(j \leq i))}$ 
giving this modality a natural presentation by sums of affines.

In \cite{gratzer2024directedunivalencesimplicialhomotopy} it is shown that $\bI$ is simplicial only by first taking a model extension of 
HoTT, and then adding two additional axioms of triangulated type theory, in addition to duality.
Using the sheaf condition established here we prove it is simplicial just using duality, and with no modal type theory.

To do so we establish some lattice theory, specifically on congruences in distributive lattices.
Given a lattice $L$, and $a,b: L$ we will write $(a = b) \subseteq L^2$ for the smallest congruence containing $(a,b)$. 
We will write $(a \leq b)$ for $(a \vee b = b)$.
\begin{lemma}
    $(a = b) = (a \wedge b =  a \vee b)$
\end{lemma}

\begin{theorem}[\cite{gratzer2009lattice} Chapter 9, Theorem 3]
    \label{th:principal_cong_in_distrib_lattice}
    Let $L$ be a distributive lattice, and $a,b,x,y : L$ with $a \leq b$.
    Then $x \equiv y \pmod{a = b}$ iff $x \wedge a = y \wedge a$ and $x \vee b = y \vee b$.
\end{theorem}

\begin{corollary}
    Let $L$ be a distributive lattice and $a, b : L$.
    Then $L/(a = b) = 0$ iff $a$ and $b$ are complements.
\end{corollary}
\begin{proof}
    We note first that $L/(a = b) = L/(a \wedge b = a \vee b)$, and that clearly $a \wedge b \leq a \vee b$.
    Thus we have $L/(a = b) = 0$ if and only if $0 \equiv 1 \pmod{a \wedge b = a \vee b}$.
    By \cref{th:principal_cong_in_distrib_lattice}, this is equivalent to  
    \[ 0 =  0 \wedge (a \wedge b) = 1 \wedge (a \wedge b) = a \wedge b \]
    \[ a \vee b =  0 \vee ( a \vee b) = 1 \vee (a \vee b) = 1\]
    So $a$ and $b$ are complements.
\end{proof}

\begin{lemma}
    For all $a, b : \bI$ we have $\neg (a = b) = (a \vee b = 1) \times (a \wedge b = 0)$.
\end{lemma}
\begin{proof}
    Note that $\neg (a = b)$ is equivalent to $\neg \Spec(\bI/(a = b))$.
    By duality that is equivalent to $\bI/(a=b)$ being the zero algebra, which is equivalent to $a \wedge b = 0$ and $a \vee b = 1$.
\end{proof}

This further implies $\neg (a = 0) = (a = 1)$ and $\neg (a = 1) = (a = 0)$.

\begin{lemma}
    $\bI$ is simplicial.
\end{lemma}
\begin{proof}
    By the sheaf condition, \cref{cor:sheaf_cond}, $\bI$ is simplicial if and only if for all $i, j : \bI$ we have
    \[ \bI \simeq \lim( \bI/(i \leq j) \times \bI/(j \leq i) \rightrightarrows \bI/(i \leq j) \times \bI/(i = j) \times \bI/(j \leq i) )\]
    Simplifying this further it is equivalent to 
    \[ \bI \simeq \lim ( \bI/(i \leq j) \times \bI/(j \leq i) \rightrightarrows \bI/(i = j) )\]
    That is given $w_1 : \bI/(i \leq j)$ and $w_2 : \bI/(j \leq i)$, whose projections into $\bI/(i = j)$ are equal,
    there is a unique $z : \bI$ projecting onto $w_1$ and $w_2$.
    We prove this now.

    We start with uniqueness. Suppose $z_1, z_2 : \bI$ both project onto $w_1$ and $w_2$.
    Then we have $z_1 \equiv z_2 \pmod{j = i \vee j}$ and $z_1 \equiv z_2 \pmod{ i = i \vee j }$.
    Then by \cref{th:principal_cong_in_distrib_lattice} we have 
        \[ z_1 \wedge j = z_2 \wedge j \]  
        \[ z_1 \vee i \vee j = z_2 \vee i \vee j \] 
        \[ z_1 \wedge i = z_2 \wedge i \] 
    Combining using distributivity further we have 
    \[ z_1 \wedge (i \vee j) = z_2 \wedge (i \vee j)\]
    and so we can see that $z_1, z_2$ are complements of $i \vee j$ in the interval $[z_1 \wedge (i \vee j), z_1 \vee i \vee j] \subseteq \bI$.
    But complements are unique in distributive lattices so $z_1 = z_2$.

    Now suppose we have $w_1, w_2$ as in the setup.
    Write $w_1 = [x]$ and $w_2 = [y]$ for some $x, y : \mathbb{I}$.
    First we note that since $x \equiv y \pmod{i = j}$, and noting that this is equivalent to $x \equiv y \pmod{i \wedge j = i \vee j}$, by \cref{th:principal_cong_in_distrib_lattice} we deduce that
    \[ x \wedge i \wedge j = y \wedge i \wedge j \]
    \[ y \vee i \vee j = y \vee i \vee j \] 
    We choose $z = (x \vee y) \wedge (x \vee i) \wedge (y \vee j)$.
    Then we have
    \begin{align*}
        z &= (x \vee y) \wedge (x \vee i) \wedge (y \vee j) \\
          &\equiv (x \vee y) \wedge (x \vee i) \wedge (y \vee i \vee j) \pmod{i \leq j} \\
          &= (x \vee y) \wedge (x \vee i) \wedge (x \vee i \vee j) \\
          &\equiv (x \vee y) \wedge (x \vee j) \wedge (x \vee j) \pmod{i \leq j} \\
          &= x \vee (y \wedge i \wedge j) \\
          &= x \vee (x \wedge i \wedge j) \\
          &= x
    \end{align*}
    so that $z \equiv x \pmod{i \leq j}$.
    Dually, swapping $x,y$ and $i,j$ shows that $z \equiv y \pmod{j \leq i}$ and hence we are done.
\end{proof}

\begin{corollary}
    For every finitely presented $\mathbb{I}$-algebra $A$, the type $\Spec(A)$ is simplicial.
\end{corollary}

\begin{remark}
    Note that the previous result did not rely on the projectivity axiom, or the constancy of maps to the natural numbers, since we do not need that we have a presentation of a modality here.
    So this result really does reduce the number of dependencies required to show $\mathbb{I}$ is simplicial. 
\end{remark}

The generators of this presentation as written are not enough to provide a description of covers, since they aren't closed under sums.
To rectify this we extend the presentation to have chains of elements of length $n$.
For an inequality $(a \leq b)$ we write $(a \leq b)^{-1} = b \leq a$
Then we let $\mathrm{Simp}$ be the presentation generated by 
\[ \sum_{\sigma : \mathrm{Fin}(n) \to \{-1, 1\}} \Spec( \mathbb{I} / (i_1 \leq j_1)^{\sigma(1)}, \ldots, (i_n \leq j_n)^{\sigma(n)})\] 
for $i_1, \ldots, i_n, j_1, \ldots, j_n : \mathbb{I}$
Then these are closed under $\sigma$ via an analogous argument to synthetic algebraic geometry.

\begin{lemma}
    Every cover in the simplicial topology of $\Spec(A)$ is given by 
    \[\sum_{\sigma : \mathrm{Fin}(n) \to \{-1, 1\}} \Spec(A / ((a_1 \leq b_1)^{\sigma(1)}, \ldots, (a_n \leq b_n)^{\sigma(n)}) ) \]
    where $a_i, b_i : A$ and $\sigma : \mathrm{Fin}(n) \to \{ -1, 1\}$.
\end{lemma}

We thus obtain a corresponding simplicial local choice axiom:
\begin{lemma}
    For any finitely presented $\bI$-algebra $A$ and any simplicially merely inhabitted type family $B : \Spec(A) \to \UU$, 
    there are $a_i, b_i : A$ and choices of sections 
    \[s_{\sigma} : (x : \Spec(A / ((a_1 \leq b_1)^{\sigma(1)}, \ldots, (a_n \leq b_n)^{\sigma(n)}))  ) \to B(x)\] 
    for all $\sigma : \mathrm{Fin}(n) \to \{ -1, 1\}$. 
\end{lemma}

\section{Future work and Conclusion}

We have continued work in extending the ways homotopy type theory can be used to study higher topos theory, specifically in creating useful tools for creating topological modalities.
Establishing presentaitons allowed us to show extensions of classical topos theoretic principles, specifically in the form of higher sheaf conditions.
Using principles that hold true in presheaf topoi, we define projective presentations, and show these can be used to explicitly calculate sheafification of propositions.
We also show how choice principles descend to subuniverses.
Finally, we apply this to the example of a classifying category for algebraic theories.
We specialise to show results about triangulated type theory and synthetic algebraic geometry.

In the future, we hope that homotopy type theory and related type theories will continue to improve as places in which to perform synthetic mathematics.
We see the axioms established in section 6 as being generic enough to be applicable to many algebraic settings, and we hope that both duality and local-choice results get use in 
other ares of synthetic mathematics.
A possible future direction is to finding a setting in which these axioms compute, or are derivable in a computational setting, perhaps by restricting the type of theory further.
With work being done on introducing tiny objects into HoTT, we also hope that we can consider presentations generated by tiny objects and gain full descriptions of sheafification,
rather than just for propositions, among other computational tools.

\printbibliography

\end{document}